\documentclass[submission,copyright,creativecommons]{eptcs}
\providecommand{\event}{TARK  2021} 
\usepackage{breakurl}             
\usepackage{underscore}           

\title{Attainable Knowledge and Omniscience}
\author{Pavel Naumov
\institute{King's College\\ Pennsylvania, United States}
\email{pgn2@cornell.edu}
\and
Jia Tao
\institute{Lafayette College\\
Pennsylvania, United States}
\email{taoj@lafayette.edu}
}
\def\titlerunning{Attainable Knowledge and Omniscience}
\def\authorrunning{Pavel Naumov \& Jia Tao}

\usepackage{amsmath}
\usepackage{wasysym}
\usepackage{amssymb}
\usepackage{amsmath}
\usepackage{graphicx}

\newtheorem{theorem}{Theorem}
\newtheorem{lemma}{Lemma}

\newtheorem{definition}{Definition}

\renewcommand{\phi}{\varphi}
\renewcommand{\epsilon}{\varepsilon}

\renewcommand{\>}{\rangle}
\newenvironment{proof}{\noindent{\sc Proof.}}{\hfill $\boxtimes\hspace{2mm}$\linebreak}

\newenvironment{proof-of-claim}{\noindent{\sc Proof of Claim.}}{\hfill $\boxtimes\hspace{2mm}$\linebreak}

\renewcommand{\phi}{\varphi}
\newcommand{\K}{{\sf K}}
\newcommand{\B}{{\sf B}}
\renewcommand{\boxdot}{{\sf A}}
\renewcommand{\Box}{{\sf O}}

\begin{document}
\maketitle

\begin{abstract}
The paper investigates an evidence-based semantics for epistemic logics. It is shown that the properties of knowledge obtained from a potentially infinite body of evidence are described by modal logic S5. At the same time, the properties of knowledge obtained from only a finite subset of this body are described by modal logic S4. The main technical result is a sound and complete bi-modal logical system that describes properties of these two modalities and their interplay.
\end{abstract}

\maketitle

\section{Introduction}


The distinction between {\em potential} and {\em actual} infinity goes back to Zeno's paradoxes. The former refers to constructions involving a potentially unlimited, but finite, number of steps while the latter refers to an infinite set as a whole.

There has been a long tradition of distinguishing these two types of infinity in logic. We say that a set is decidable if the membership in the set can be verified through a potentially unlimited but finite number of steps of a Turing machine. Most people believe that Peano Arithmetic is consistent and, thus, its consistency could be verified by observing that none of the  countably many finite-length derivations ends with a contradiction. Yet, G\"{o}del's second incompleteness theorem claims that the consistency cannot be established by a single finite-length derivation or a finite set of such derivations. More generally, logicians usually accept admissible inference rules with finitely many assumptions and reject $\omega$-rules that deduce statements from infinitely many assumptions. A compactness property for a logical system states that if a statement semantically follows from an infinite set of assumptions, then it semantically follows from a potentially unlimited but finite subset of the assumptions.

In this paper we investigate the difference between the knowledge that can be derived from a potentially unlimited, but finite, subset of the body of  evidence and the knowledge that can be inferred from the whole infinite body of evidence. We refer to the first form of knowledge as {\em attainable knowledge} and to the second as {\em omniscience}. We show that the properties of the attainable knowledge are captured by modal logic S4, while those of the omniscience are given by modal logic S5. 

\subsection{Epistemic Logics S4 and S5}\label{S4 S5 subsection}

Epistemic logic S5 is an extension of logic S4 by the Negative Introspection axiom: $\neg\K\phi\to\K\neg\K\phi$, where $\K$ is the knowledge modality. Informally, this axiom expresses an agent awareness of the limits of her knowledge: if an agent does not know something, then she must know that she does not know it. Even without taking into account the size of the body of supporting evidence, the validity of this axiom has long been a subject for philosophical discussions. Hintikka~\cite{h62} rejects this axiom. Stalnaker~\cite{s06ps} cites a belief-based argument against it. The argument relies on three standard assumptions about knowledge and beliefs: (i) beliefs of an agent are consistent, (ii) an agent believes in anything she knows, and (iii) anything an agent knows is true.
We interpret Stalnaker's argument in terms of a mathematician's belief about her knowledge. Imagine a mathematician who spent last several years working on a conjecture $\phi$. One day she announces ``I believe I now know that the conjecture is true''\footnote{In Stalkner's words, ``agent \dots take[s] herself to know'' that the conjecture is true~\cite[p.177]{s06ps}.}. We write this as $\B\K\phi$. Because her beliefs are consistent, we must conclude that she does not believe in the opposite:  $\neg\B\neg\K\phi$. Hence, $\neg\K\neg\K\phi$ by the contraposition of assumption (ii). Thus, by the contrapositive of the Negative Introspection axiom, we have $\K\phi$. Then, $\phi$ is true by assumption (iii). Following Stalnaker, we have just shown that if a mathematician believes that she knows that conjecture $\phi$ is true, then conjecture $\phi$ must be true. 
However, intuitively, this cannot be true because her belief alone should not make the conjecture true.

While the Negative Introspection axiom might not be true as the most general epistemic axiom, Fagin, Halpern,  Moses, and Vardi point out that it is ``appropriate for many multi-agent systems applications'', such as the analysis of the muddy children puzzle, knowledge bases, games of imperfect information, and message-passing systems~\cite[p.18]{fhmv95}. In spite of its limitations, S5 has been accepted as the standard epistemic logic in the modern literature~\cite{aa16jlc,abvs10jal,dhk07,mv04}. There also have been suggestions to consider modal systems in between S4 and S5
~\cite{k76}, 
\cite[p.82]{l78}.

The purpose of this paper is not to settle the discussion about the merits of S4 vs. S5, but to highlight the difference between these two systems from the viewpoint of an evidence-based semantics. The main technical result of our work is a sound and complete logical system with two modalities that capture individual properties of the attainable knowledge and of the omniscience as well as the properties that connect these two types of knowledge.

\subsection{Grand Hotel Example} 

Consider the famous Hilbert Grand Hotel that has infinitely many rooms. By a single observation in this setting we mean examining a room in order to establish whether it is empty. If the hotel has vacancies, then at least one room is empty. By opening just a single door of that specific room one can learn that the hotel has vacancy. We write $\boxdot(\mbox{``the hotel has vacancy"})$ to express this fact.
The modality $\boxdot$ denotes {\em attainable knowledge} that can be formed from finitely many observations.  In fact, in this case a single observation suffices. On the other hand, if the hotel has no vacancy, then the knowledge about this cannot be obtained from examining any finite subset of the rooms: $\neg\boxdot(\mbox{``the hotel has no vacancy"})$. At the same time, one can learn that the hotel is full by examining all rooms in the hotel. We write this as $\Box(\mbox{``the hotel has no vacancy"})$, where modality $\Box$ denotes the {\em omniscience} that can be formed from a potentially infinite body of evidence. 

Let us first show that attainable knowledge does not satisfy the Negative Introspection axiom:
\begin{equation}\label{attainable negative introspection}
    \neg\boxdot\phi\to\boxdot\neg\boxdot\phi.
\end{equation}
Consider the epistemic world $w$ in which the hotel is full and let $\phi$ be statement ``the hotel has vacancy". Thus, $w\nVdash\phi$. Hence, one cannot know that $\phi$ is true from examining (finitely many) hotel rooms: $w\nVdash\boxdot\phi$. Hence, $w\Vdash\neg\boxdot\phi$. At the same time, informally, the only reason why $w\Vdash\neg\boxdot\phi$ is true is because the hotel is full. Since the latter cannot be established through finitely many observations,  $w\Vdash\neg\boxdot\neg\boxdot\phi$. More formally, to prove $w\Vdash\neg\boxdot\neg\boxdot\phi$ we need to show that no matter which finite set of evidence is examined, there still will be an epistemic world $u$ indistinguishable from $w$ such that  $u\Vdash\boxdot\phi$. Indeed, suppose that we have chosen to examine rooms $r_1,\dots,r_n$. Let $r$ be any room different from rooms $r_1,\dots,r_n$. Let $u$ be an epistemic world in which all rooms except for room $r$ are occupied. Thus, epistemic worlds $w$ and $u$ are indistinguishable through the chosen finite set of examinations. Yet, $u\Vdash\boxdot\phi$ because in epistemic world $u$, it is enough to examine a single room $r$  to learn that the hotel has vacancy. This concludes the counterexample for the Negative Introspection axiom~(\ref{attainable negative introspection}). 

From the classical (non-evidence-based) Kripke semantics point of view, the Negative Introspection axiom captures the fact that the indistinguishability relation on epistemic states is symmetric. In our evidence-based semantics, the indistinguishability by a single evidence is also symmetric. That is, if one can distinguish one state of the hotel from another state by examining just a single room, then one can also distinguish the latter state from the former by examining the same room. The reason why the Negative Introspection axiom fails, under the evidence-based semantics, is because {\em the indistinguishability relation between a world and a set of worlds by a potentially unlimited, but finite body of evidence is not symmetric.} Indeed, consider the epistemic world in which the hotel is empty and the set of all worlds in which the hotel is not empty. By examining a right single room in a non-empty hotel, one can distinguish it from an empty hotel. Yet, after examining any finite subset of rooms in an empty hotel one cannot distinguish the current world from the set of worlds in which the hotel is not empty. 

Let us now consider a weaker form of the Negative Introspection axiom used in logic S4.4~\cite{k76}:
\begin{equation}\label{weak attainable negative introspection}
    \phi\to(\neg\boxdot\phi\to\boxdot\neg\boxdot\phi).
\end{equation}
The above counterexample for formula~(\ref{attainable negative introspection}) does not work as a counterexample for formula~(\ref{weak attainable negative introspection}) because in that setting $w\nVdash\phi$. Nevertheless, principle  (\ref{weak attainable negative introspection}) is not valid in general. Indeed, let us modify the hotel setting by assuming that some rooms in the hotel might have bedbugs. The presence of the bedbugs can be tested when a room is examined. Furthermore, let us assume that once a single room in the hotel becomes infected with bedbugs all guests immediately leave the hotel. Thus, the hotel might have either (a) visitors and no bedbugs, or (b) no visitors and no bedbugs, or (c) bedbugs and no visitors. Consider an epistemic world $w$ in which the hotel has no visitors and no bugs. Let $\phi$ be the statement ``the hotel has no visitors". Thus, $w\Vdash \phi$. Note that there are two ways to verify that the hotel has no visitors: either by examining all rooms to observe that they are all vacant or by examining a single room that contains bedbugs. Since in the epistemic world $w$ the hotel is not infected with bugs, the only way to verify that the hotel is empty in this world is to examine all rooms. Thus,  $w\Vdash\neg\boxdot\phi$. To finish the counterexample for formula~(\ref{weak attainable negative introspection}), it suffices to show that  $w\nVdash\boxdot\neg\boxdot\phi$. In other words, we need to prove that, after examining a finite set of rooms $r_1,\dots,r_n$, we will not be able to distinguish epistemic state $w$ from such an epistemic state $u$ that $u\Vdash\boxdot\phi$. Indeed, let $r$ be any room different from rooms $r_1,\dots,r_n$ and let $u$ be the epistemic world in which room $r$ is infected with bedbugs. Note that $u\Vdash\boxdot\phi$ because in epistemic state $u$ it is enough to examine the bug-infected room $r$ to conclude that the hotel is empty. 

\subsection{Formal Semantics of Evidence}  

To make our Grand Hotel example more formal, we introduce the formal semantics of evidence-based knowledge. Since we view evidence as a way to distinguish epistemic worlds, in this paper we interpret the pieces of evidence as equivalence relations on the worlds. When an agent takes into account several pieces of evidence, equivalence relations corresponding to these pieces intersect to form the equivalence relation of the agent. In the second Grand Hotel example above, each room is in one of the three states: occupied, vacant without bedbugs, and vacant with bedbugs. An epistemic world can be described by specifying the state of each room. The observation that consists of examining room $r$ can distinguish two epistemic worlds in which room $r$ is in different states.  It cannot distinguish two epistemic worlds in which room $r$ is in the same state. 

In other words, we assemble an agent's knowledge from several pieces of evidence in a similar way as how distributed knowledge~\cite{fhmv95} of a group is assembled from the individual observations of the members of the group. Thus, the logical system developed in this paper could also be used to describe two different forms of group knowledge by an infinite group of agents: modality $\Box$ represents the standard distributed knowledge by the whole group and modality $\boxdot$ represents the distributed knowledge by a finite subgroup of the whole group.

The formal evidence-based semantics described above is similar to the one for the budget-constrained knowledge proposed by 
us earlier~\cite{nt15aamas,nt17jal}. It is different from the neighbourhood semantics of evidence investigated by van Benthem and Pacuit~\cite{bp11sl} and the probabilistic approach of Halpern and Pucella~\cite{hp06jair}. It is not clear how and if the results in the current paper could be applied to these alternative semantics.

\subsection{Logical System}\label{logical system subsection}

In this paper we propose a sound and complete logical system that describes the universal properties of omniscience modality $\Box$ and attainable knowledge modality $\boxdot$. The axioms involving only modality $\Box$ are exactly those forming modal logic S5:
\begin{enumerate}
    \item  Truth: $\Box\phi\to\phi$,
    \item  Positive Introspection: $\Box\phi\to\Box\Box\phi$, 
    \item  Negative Introspection: $\neg\Box\phi\to\Box\neg\Box\phi$,
    \item  Distributivity:  $\Box(\phi\to\psi)\to(\Box\phi\to\Box\psi)$. 
\end{enumerate}
Later we do not list the Positive Introspection axiom among axioms of our system because, just like in the case of logic S5, the Positive Introspection axiom is derivable from the other axioms. We prove this in Lemma~\ref{potential positive introspection lemma}.
The axioms involving only modality $\boxdot$ are exactly those forming modal logic S4:
\begin{enumerate}
    \item  Truth: $\boxdot\phi\to\phi$,   
    \item  Positive Introspection: $\boxdot\phi\to\boxdot\boxdot\phi$,
    \item Distributivity: $\boxdot(\phi\to\psi)\to(\boxdot\phi\to\boxdot\psi)$.
\end{enumerate}
Finally, there appears to be two independent axioms that capture the interplay of the two modalities:
\begin{enumerate}
    \item Monotonicity: $\boxdot\phi\to\Box\phi$,
    \item Mixed Negative Introspection: $\neg\boxdot\phi\to\Box\neg\boxdot\phi$.
\end{enumerate}
The first of these axioms states that any knowledge that can be formed from a finite subset  of observations can also be formed on the bases of the whole set of observations. The second axiom is significantly more interesting. We have seen in our first Grand Hotel example that the Negative Introspection axiom is not true for attainable knowledge. Namely an agent in a fully occupied hotel cannot learn that statement $\boxdot(\mbox{``the hotel has vacancies"})$ is false by examining only finitely many rooms. It can learn this, however, by examining all room in the hotel and this is exactly what the Mixed Negative Introspection axiom claims. Surprisingly, the Mixed Negative Introspection axiom is provable from other axioms of our logical system. We show this in Lemma~\ref{mixed negative introspection lemma}. As a result, Mixed Negative Introspection is not included as an axiom of our system. Additionally, it is easy to see that the Truth axiom for modality $\boxdot$ follows from the Truth axiom for modality $\Box$ and the Monotonicity axiom. For this reason, we do not list the Truth axiom for modality $\boxdot$ as one of our axioms either. Finally, although one can state two forms of the Necessitation inference rule: one for  modality $\Box$ and another for modality $\boxdot$, the former follows from the latter and the Monotonicity axiom. Thus, our system, in addition to Modus Ponens, only includes the Necessitation rule for modality $\boxdot$.

\subsection{Outline}

The rest of the paper is organized as following. Section~\ref{syntax and semantics section} defines the syntax and an evidence-based semantics of our logical system. Section~\ref{axioms section} lists the axioms and  the inference rules of this system 
and proves
the Mixed Negative Introspection axiom mentioned above. 
In Section~\ref{soundness section}, we prove the soundness of our logical system with respect to the evidence-based semantics. In Section~\ref{canonical model section}, we state the completeness theorem and define the canonical model used in its proof. The full proof of the completeness can be found in the appendix. Section~\ref{conclusion section} concludes. 

\section{Syntax and Semantics}\label{syntax and semantics section}

In this section, we describe the formal syntax and the formal evidence-based semantics of our logical system. Throughout the rest of the paper we assume a fixed infinite set of propositional variables. 

\begin{definition}
The set $\Phi$ of all formulae $\phi$ is defined by grammar
$
\phi := p \;|\; \neg \phi \;|\; \phi\to\phi \;|\; \boxdot \phi \;|\; \Box \phi,
$
where $p$ represents propositional variables.
\end{definition}

\begin{definition}
An evidence model is $\<W,E,\{\sim_{e}\}_{e\in E},\pi\>$, where
\begin{enumerate}
    \item $W$ is a (possibly empty) set of ``epistemic worlds",
    \item $E$ is an arbitrary (possibly empty) ``evidence" set,
    \item $\sim_e$ is an ``indistinguishability" equivalence relation on $W$ for each $e\in E$,
    \item $\pi$ is a function that maps propositional variables into subsets of $W$.
\end{enumerate}
\end{definition}
For instance, in our first Grand Hotel example, an epistemic world is a function $\mathbb{N}\to \{vacant, occupied\}$ that assigns a state to each room in the hotel. The set of evidence is $\mathbb{N}$, where evidence with number $r\in \mathbb{N}$ corresponds to examining room with number $r$ in this hotel. Epistemic worlds $w_1$ and $w_2$ are $\sim_r$-equivalent if room number $r$ has the same state in both of the worlds. In other words, $w_1\sim_r w_2$ if $w_1(r)=w_2(r)$. Finally a function $\pi$ may, for example, map propositional variable $p$ into the set of all epistemic worlds representing a nonempty hotel: $\pi(p)=\{w\in W\;|\; \exists r\in \mathbb{N}\mbox{ such that } w(r)=occupied\}.$

In this paper, we write $w_1\sim_F w_2$ if $w_1\sim_e w_2$ for each $e\in F$. In particular, $w_1\sim_\varnothing w_2$ for any two epistemic worlds $w_1,w_2\in W$. 
\begin{definition}\label{sat}
For any formula $\phi\in\Phi$ and any epistemic world $w\in W$ of an evidence model $\<W,E,\{\sim_e\},\pi\>$, let the satisfaction relation $w\Vdash\phi$ be defined as follows,
\begin{enumerate}
    \item $w\Vdash p$, if $w\in \pi(p)$,
    \item $w\Vdash\neg\phi$, if $w\nVdash\phi$,
    \item $w\Vdash\phi\to\psi$, if $w\nVdash\phi$ or $w\Vdash\psi$,
    \item $w\Vdash\boxdot\phi$, if there is a finite $F\subseteq E$ such that  $w\sim_{F}u$ implies $u\Vdash\phi$ for each $u\in W$,
    \item $w\Vdash\Box\phi$, if $u\Vdash\phi$ for each $u\in W$ such that $w\sim_{E}u$.
\end{enumerate}
\end{definition}

\section{Axioms}\label{axioms section}

In addition to propositional tautologies in language $\Phi$, our system consists of the following axioms:

\begin{enumerate}
    \item  Truth: $\Box\phi\to\phi$,
        \item  Positive Introspection: $\boxdot\phi\to\boxdot\boxdot\phi$,
    \item  Negative Introspection: $\neg\Box\phi\to\Box\neg\Box\phi$,
    \item  Distributivity: 
    $\boxdot(\phi\to\psi)\to(\boxdot\phi\to\boxdot\psi)$ and
    $\Box(\phi\to\psi)\to(\Box\phi\to\Box\psi)$,
    \item Monotonicity: $\boxdot\phi\to\Box\phi$.

\end{enumerate}

We say that a formula $\phi$ is a theorem in our logical system and write $\vdash\phi$ if formula $\phi$ is derivable from the axioms of our systems using  the Modus Ponens and the Necessitation inference rules:
$$
\dfrac{\phi,\;\;\; \phi\to\psi}{\psi}
\hspace{15mm}
\dfrac{\phi}{\boxdot\phi}.
$$
We write $X\vdash\phi$ if formula $\phi$ is derivable from the theorems of our logical systems and an additional set of axioms $X$ using {\em only} the Modus Ponens inference rule. We say that set $X\subseteq\Phi$ is inconsistent if there is a formula $\phi\in\Phi$ such that $X\vdash\phi$ and $X\vdash\neg\phi$.
%

\begin{lemma}\label{jan29-a}
If $X$ is a finite consistent set of formulae, and $p$ is a propositional variable that does not occur in $X$, then sets $X\cup\{p\}$ and $X\cup\{\neg p\}$ are both consistent. 
\end{lemma}
\begin{proof}
Suppose that either set $X\cup\{p\}$ or set $X\cup\{\neg p\}$ is inconsistent. Thus, either $X\vdash \neg p$ or $X\vdash p$. 

First, we show that $X\vdash \neg p$ in both cases. Indeed, suppose that $X\vdash p$. Let $\phi_1,\dots,\phi_n$ be the proof of formula $p$ from the set of the assumption $X$. Let $\phi_i'$ be the result of the replacement of propositional variable $p$ by formula $\neg p$ in formula $\phi_i$. Note that variable $p$ does not occur in set $X$ by the assumption of the lemma. Then, $\phi'_1,\dots,\phi'_n$ is the proof of formula $\neg p$ from the set of the assumption $X$. Therefore, $X\vdash \neg p$. 

Next, we show that $X\vdash \neg\neg p$. 
Let $\psi_1,\dots,\psi_m$ be the proof of formula $\neg p$ from the set of the assumption $X$. Consider now sequence $\psi'_1,\dots,\psi'_m$ obtained by replacing each occurrence of variable $p$ with $\neg p$ in proof $\psi_1,\dots,\psi_m$. Note again that variable $p$ does not occur in set $X$ by the assumption of the lemma. Then, sequence $\psi'_1,\dots,\psi'_m$ is a proof of formula $\neg\neg p$ from the set of the assumption $X$. Therefore, $X\vdash \neg\neg p$. 

Statements $X\vdash \neg p$ and $X\vdash \neg\neg p$ imply that set $X$ is not consistent.
\end{proof}

\begin{lemma}\label{potential positive introspection lemma}
$\vdash \Box\phi\to\Box\Box\phi$.
\end{lemma}
\begin{proof}
Note that formula $\Box\neg\Box\phi\to\neg\Box\phi$ is an instance of the Truth axiom. Thus, $\vdash \Box\phi\to\neg\Box\neg\Box\phi$ by the law of contrapositive in the propositional logic. Hence, taking into account the following instance of  the Negative Introspection axiom $\neg\Box\neg\Box\phi\to\Box\neg\Box\neg\Box\phi$,
one can conclude that 
\begin{equation}\label{pos intro eq}
\vdash \Box\phi\to\Box\neg\Box\neg\Box\phi.
\end{equation}

At the same time,  $\neg\Box\phi\to\Box\neg\Box\phi$ is an instance of the Negative Introspection axiom. Thus, $\vdash \neg\Box\neg\Box\phi\to \Box\phi$ by contraposition. Hence, by the  Necessitation inference rule,
$\vdash \boxdot(\neg\Box\neg\Box\phi\to \Box\phi)$. Thus, $\vdash \Box(\neg\Box\neg\Box\phi\to \Box\phi)$ by the Monotonicity axiom and the Modus Ponens inference rule. Thus, by  the Distributivity axiom and the Modus Ponens inference rule, 
$\vdash \Box\neg\Box\neg\Box\phi\to \Box\Box\phi$.
The last statement, together with statement~(\ref{pos intro eq}), implies the statement of the lemma by the laws of propositional reasoning.
\end{proof}

We conclude this section with the proof of the Mixed Negative Introspection axiom mentioned in Section~\ref{logical system subsection} of the introduction.

\begin{lemma}\label{mixed negative introspection lemma}
$\vdash\neg\boxdot\phi\to\Box\neg\boxdot\phi$.
\end{lemma}
\begin{proof}
By the  Truth axiom, $\vdash \Box\boxdot\phi\to\boxdot\phi$. Thus, by the law of contrapositive in the propositional logic, $\vdash \neg\boxdot\phi\to\neg\Box\boxdot\phi$. At the same time, by  the Negative Introspection axiom, $\vdash \neg\Box\boxdot\phi\to\Box\neg\Box\boxdot\phi$. Hence, by the laws of propositional reasoning,
\begin{equation}\label{interplay eq}
    \vdash \neg\boxdot\phi\to\Box\neg\Box\boxdot\phi.
\end{equation}

Note that $\vdash \boxdot\phi\to\boxdot\boxdot\phi$ by the  Positive Introspection axiom and $\vdash \boxdot\boxdot\phi\to \Box\boxdot\phi$ by the Monotonicity axiom. Thus, by the laws of propositional reasoning,
$\vdash\boxdot\phi\to\Box\boxdot\phi$. Hence, by the law of contrapositive in the propositional logic, $\vdash \neg\Box\boxdot\phi\to\neg\boxdot\phi$. Then, $\vdash \boxdot(\neg\Box\boxdot\phi\to\neg\boxdot\phi)$ by the  Necessitation inference rule. Thus, $\vdash \Box(\neg\Box\boxdot\phi\to\neg\boxdot\phi)$ by the Monotonicity axiom and the Modus Ponens inference rule. Hence, $\vdash \Box\neg\Box\boxdot\phi\to\Box\neg\boxdot\phi$ by  the Distributivity axiom and the Modus Ponens inference rule. Therefore, $\vdash\neg\boxdot\phi\to\Box\neg\boxdot\phi$, by the laws of propositional reasoning taking  into account statement~(\ref{interplay eq}).
\end{proof}


\section{Soundness}\label{soundness section}



\begin{theorem}[strong soundness]\label{strong soundness}
For any world $w\in W$ of each evidence model $\<W,E,\{\sim_e\}_{e\in E},\pi\>$, any set of formulae $X\subseteq \Phi$, and any formula $\phi\in\Phi$, if $w\Vdash\chi$ for each formula $\chi\in X$ and $X\vdash\phi$, then $w\Vdash \phi$.
\end{theorem}

The soundness of propositional tautologies and Modus Ponens inference rule is straightforward. Below we prove the soundness of each of the remaining axioms and the  Necessitation inference rule as separate lemmas.

\begin{lemma}
If $w\Vdash\Box\phi$, then $w\Vdash\phi$.
\end{lemma}
\begin{proof}
By Definition~\ref{sat},
assumption $w\Vdash\Box\phi$ implies that $u\Vdash\phi$ for all $u\in W$ such that $w\sim_E u$. Note that  $w\sim_e w$ for all $e\in E$ because $\sim_e$ is an equivalence relation. Hence, $w\sim_E w$. Therefore, $w\Vdash\phi$.
\end{proof}

\begin{lemma}
If $w\Vdash\boxdot\phi$, then $w\Vdash\boxdot\boxdot\phi$.
\end{lemma}
\begin{proof}
By Definition~\ref{sat}, the assumption $w\Vdash\boxdot\phi$ implies that there is a finite set $F\subseteq E$ such that $u\Vdash\phi$ for each $u\in W$ where $w\sim_F u$.

Again by Definition~\ref{sat}, it suffices to show that $v\Vdash\boxdot\phi$ for all $v\in W$ such that $w\sim_F v$. To establish this, it is enough to prove that $u\Vdash\phi$ for all $u\in W$ such that $v\sim_F u$. Note that $w\sim_F v\sim_F u$. Thus, $w\sim_F u$ because $\sim_f$ is an equivalence relation for each element $f\in F$. Then, $u\Vdash\phi$ by the choice of set $F$.
\end{proof}

\begin{lemma}
If $w\Vdash\neg\Box\phi$, then $w\Vdash\Box\neg\Box\phi$.
\end{lemma}
\begin{proof}
By Definition~\ref{sat},
assumption  $w\Vdash\neg\Box\phi$ implies that there is $u\in W$ such that $w\sim_E u$ and $u\nVdash \phi$.

Consider any $v\in W$ such that $w\sim_E v$. By Definition~\ref{sat}, to prove $w\Vdash\Box\neg\Box\phi$,  it suffices to show that $v\Vdash \neg\Box\phi$. Note that $w\sim_E u$ and $w\sim_E v$. Thus, $v\sim_E u$ due to $\sim_e$ being an equivalence relation for each $e\in E$. Recall that $u\nVdash \phi$. Hence, $v\nVdash\Box\phi$ by Definition~\ref{sat}. Therefore, $v\Vdash \neg\Box\phi$ again by Definition~\ref{sat}.
\end{proof}

\begin{lemma}
If $w\Vdash\boxdot(\phi\to\psi)$ and $w\Vdash\boxdot\phi$, then $w\Vdash\boxdot\psi$.
\end{lemma}
\begin{proof}
By Definition~\ref{sat}, the assumption $w\Vdash\boxdot(\phi\to\psi)$ implies that there is a finite set $F_1\subseteq E$ such that $u\Vdash\phi\to\psi$ for each $u\in W$ such that $w\sim_{F_1} u$. Similarly, the assumption $w\Vdash\boxdot\phi$ implies that there is a finite set $F_2\subseteq E$ such that $u\Vdash\phi$ for each $u\in W$ such that $w\sim_{F_2} u$.

Let $F=F_1\cup F_2$. It suffices to show that $u\Vdash\psi$ for each $u\in W$ such that $w\sim_{F} u$. Indeed, statement $w\sim_{F} u$ implies that $w\sim_{F_1} u$ and $w\sim_{F_2} u$. Hence, $u\Vdash\phi\to\psi$ and $u\Vdash\phi$ due to the choice of sets $F_1$ and $F_2$. Therefore, $u\Vdash\psi$ by Definition~\ref{sat}.
\end{proof}

\begin{lemma}
If $w\Vdash\Box(\phi\to\psi)$ and $w\Vdash\Box\phi$, then $w\Vdash\Box\psi$.
\end{lemma}
\begin{proof}
Consider any $u\in W$ such that $w\sim_E u$. By Definition~\ref{sat}, it suffices to show that $u\Vdash\psi$. Indeed, by Definition~\ref{sat}, the assumptions $w\Vdash\Box(\phi\to\psi)$ and $w\Vdash\Box\phi$ imply that $u\Vdash \phi\to\psi$ and $u\Vdash \phi$. Therefore, $u\Vdash \psi$, again by Definition~\ref{sat}.
\end{proof}

\begin{lemma}
If $w\Vdash\boxdot\phi$, then $w\Vdash\Box\phi$.
\end{lemma}
\begin{proof}
Consider any $u\in W$ such that $w\sim_E u$. By Definition~\ref{sat}, it suffices to prove that $u\Vdash\phi$. By the same definition, the assumption $w\Vdash\boxdot\phi$ implies that there is a finite set $F\subseteq E$ such that $v\Vdash\phi$ for each $v\in W$ such that $w\sim_F v$. Note that statement $w\sim_E u$ implies that $w\sim_F u$. Therefore, $u\Vdash\phi$.
\end{proof}

\begin{lemma}
If $w\Vdash\phi$ for each epistemic world $w\in W$ of each evidence model $\<W,E,\{\sim_e\}_{e\in E},\pi\>$, then $w\Vdash\boxdot\phi$ for each epistemic world $w\in W$ of each evidence model $\<W,E,\{\sim_e\}_{e\in E},\pi\>$. 
\end{lemma}
\begin{proof}
Consider any epistemic world $w\in W$ of an arbitrary evidence model $\<W,E,\{\sim_e\}_{e\in E},\pi\>$. By Definition~\ref{sat}, it suffices to show that there is $F\subseteq E$ such that $u\Vdash\phi$ for each $u\in W$ where $w\sim_F u$. Indeed, let $F=\varnothing$. Note that $u\Vdash\phi$ for each $u\in W$ due to the assumption of the lemma.
\end{proof}

This concludes the proof of the soundness of our logical system.

\section{Canonical Model}\label{canonical model section}

\begin{theorem}[strong completeness]\label{completeness theorem}
For any set of formulae $X\subseteq \Phi$ and any formula $\phi\in\Phi$, if $X\nvdash \phi$, then there is an evidence model $\<W,E,\{\sim_e\}_{e\in E},\pi\>$ and an epistemic world $w\in W$ such that $w\Vdash\chi$ for each $\chi\in X$ and  $w\nVdash\phi$.
\end{theorem}
As usual, the proof of a completeness theorem is using a canonical model construction. We define the canonical model in this section and give the full proof of the completeness in the appendix.


The standard proof of completeness for modal logics S4 and S5 defines worlds of a canonical model as maximal consistent sets of formulae. In our case, to specify a canonical model we need to define not only the set of worlds, but also the evidence set. In the construction that we propose, these two sets are defined to be the same set $W_\infty$, which is a set of nodes in a certain infinite forest. As a result, the proof of the completeness is significantly more involved than the completeness proofs for logics S4 and S5.

Throughout the section we use two operations on sequences. If $w$ is a sequence $(x_1,x_2,\dots,x_n)$ and $u$ is a sequence $(y_1,y_2\dots,y_m)$, then by concatenation $w\!::\!u$ of these two sequences we mean sequence $(x_1,x_2,\dots,x_n,y_1,y_2\dots,y_m)$. By head $hd(w)$ of a  nonempty sequence $w=(x_1,x_2,\dots,x_n)$ we mean element $x_n$. For example, $(a,b)\!::\!(c)=(a,b,c)$ and $hd(a,b,c)=c$.

We now define the ``canonical" evidence model $\<W_\infty,W_\infty,\{\sim_e\}_{e\in W_\infty},\pi\>$. The set of epistemic worlds in the canonical model is identical to the evidence set. One can think that all pieces of evidence related in some sense to a given epistemic world are combined together into a single evidence associated with this world. Because of this, the evidence associated with world $w$ is simply referred to as evidence $w$.

Informally, set $W_\infty$ consists of sequences of the form $(X_0,\epsilon_1,X_1,\dots,\epsilon_n,X_n)$ where $X_0,\dots,X_n$ are maximal consistent sets of formulae and each of $\epsilon_1,\dots,\epsilon_n$ is either a finite subset of $W_\infty$ or symbol $\ast$. In what follows, the case when $\epsilon $ is a finite subset of $W_\infty$ will form an $\boxdot$-accessibility relation and the case $\epsilon=\ast$ will form both $\boxdot$-accessibility and $\Box$-accessibility relations. Such sequences can be visualised, see Figure~\ref{forest figure}, as paths in an infinite collection of infinite trees whose vertices are maximal consistent sets of formulae and whose edges are labeled with the $\epsilon$'s described above.

\begin{figure}[ht]
\begin{center}
\vspace{0mm}
\scalebox{0.7}{\includegraphics{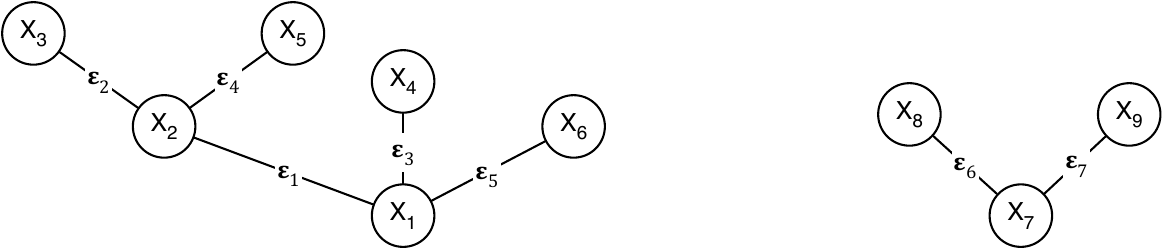}}
\vspace{0mm}
\caption{Fragments of the trees of sequences.}\label{forest figure}
\vspace{0cm}
\end{center}
\vspace{-2mm}
\end{figure}

Formally, set $W_\infty$ is specified as the union of a recursively defined infinite sequence of sets $W_0,W_1,\dots$. These sets represent different stages of building the infinite forest of infinite trees partially depicted in  Figure~\ref{forest figure}. Note, however, that stages do not correspond to the levels of the trees. Generally speaking, vertices at the same level are not created at the same stage. After a new epistemic world is added, this world can be used as a part of evidence later in the construction.

\begin{definition}\label{worlds and evidence definition}
A sequence of sets $W_0,W_1,\dots$ is defined recursively as follows.
\begin{enumerate}
    \item $W_0$ is the set of all single-element sequences $(X_0)$, where $X_0$ is a maximal consistent subset of $\Phi$,
    \item $W_{n+1}$ contains all sequences of the form $w\!::\!(\epsilon,Y)$  such that
        \begin{enumerate}
        \item $w\in \bigcup_{i\le n}W_i$,
    \item $\epsilon$ is either symbol $\ast$ or a finite subset of $\bigcup_{i\le n} W_i$,
    \item $Y$ is a maximal consistent subset of $\Phi$,
    \item if $\epsilon=\ast$, then $\{\phi\;|\; \Box\phi\in hd(w)\}\subseteq Y$, 
    \item if $\epsilon\subseteq \bigcup_{i\le n} W_i$, then $\{\phi\;|\; \boxdot\phi\in hd(w)\}\subseteq Y$. 
\end{enumerate}
\end{enumerate}
\end{definition}

\begin{definition}\label{infty definition}
$W_\infty=\bigcup_iW_i$.
\end{definition}

\begin{lemma}\label{infinity is infinite}
Set $W_\infty$ is infinite.
\end{lemma}
\begin{proof}
By Theorem~\ref{strong soundness}, all axioms of our logical system are satisfied, in particular, in the world of any single-world model with the empty evidence set. Thus, our logical system is consistent. Hence, the empty set of formulae is consistent. Since our language contains infinitely many propositional variables, the empty set could be extended to a maximal consistent set in an infinitely many ways by Lemma~\ref{jan29-a}. Hence, set $W_0$ is infinite by Definition~\ref{worlds and evidence definition}. Therefore, set $W_\infty$ is infinite by Definition~\ref{infty definition}.
\end{proof}

\begin{figure}[ht]
\begin{center}
\vspace{0mm}
\hfill
\scalebox{0.7}{\includegraphics{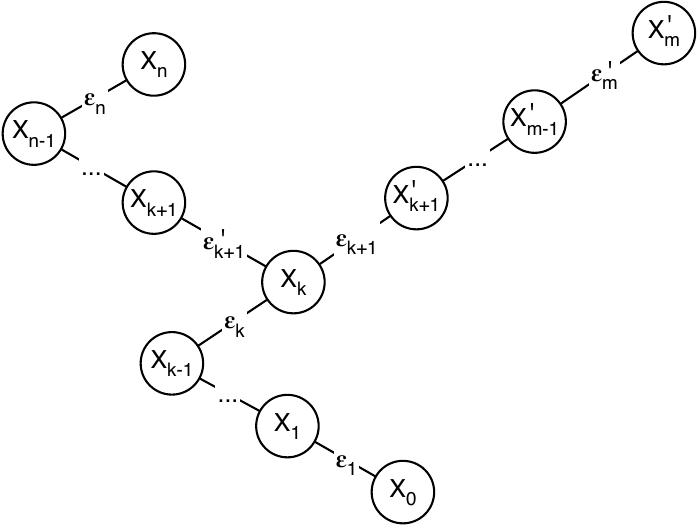}}
\hfill
\vspace{0mm}
\caption{Illustration for Definition~\ref{sim}, where $\epsilon_i$ and $\epsilon'_i$ either are equal to $\ast$ or contain $e$ for each $i\ge k$.}\label{equivalence figure (a)}
\vspace{0mm}
\end{center}
\vspace{-2mm}
\end{figure}
Informally, two sequences are $\sim_e$-equivalent if they start with the same prefix and once they deviate all subsequent $\epsilon$'s either are equal to $\ast$ or contain element $e$, see Figure~\ref{equivalence figure (a)}. The formal definition is below.
\begin{definition}\label{sim}
For any world $w=(X_0,\epsilon_1,X_1,\dots,\epsilon_n,X_n)$, any world $u=(X'_0,\epsilon'_1,X'_1,\dots,\epsilon'_m,X'_m)$, and any $e\in W_\infty$, let $w\sim_e u$ if there is $k$ such that 
\begin{enumerate}
    \item $0\le k\le \min\{n,m\}$,
    \item $X_i=X'_i$ for all $i$ such that $0\le i\le k$,
    \item $\epsilon_i=\epsilon'_i$ for all $i$ such that $0< i\le k$,
    \item for all $i$, if $k<i\le n$, then either $e\in \epsilon_i$ or $\epsilon_i=\ast$,
    \item for all $i$, if $k<i\le m$, then either $e\in \epsilon'_i$ or $\epsilon'_i=\ast$.
\end{enumerate}
\end{definition}

\begin{definition}\label{pi definition}
$\pi(p)=\{w\in W_\infty\;|\; p\in hd(w)\}$.
\end{definition}

\noindent{}The canonical evidence model $\<W_\infty,W_\infty,\{\sim_e\}_{e\in W_\infty},\pi\>$ is now fully defined.

\section{Completeness}

We start by establishing several properties of the canonical model that are necessary for our proof of the completeness.
First, we show that if set $hd(w)$ contains a formula $\boxdot\phi$, then so does set $hd(u)$ for each descendant $u$ of vertex $w$.

\begin{figure}[ht]
\begin{center}
\vspace{0mm}

\scalebox{0.7}{\includegraphics{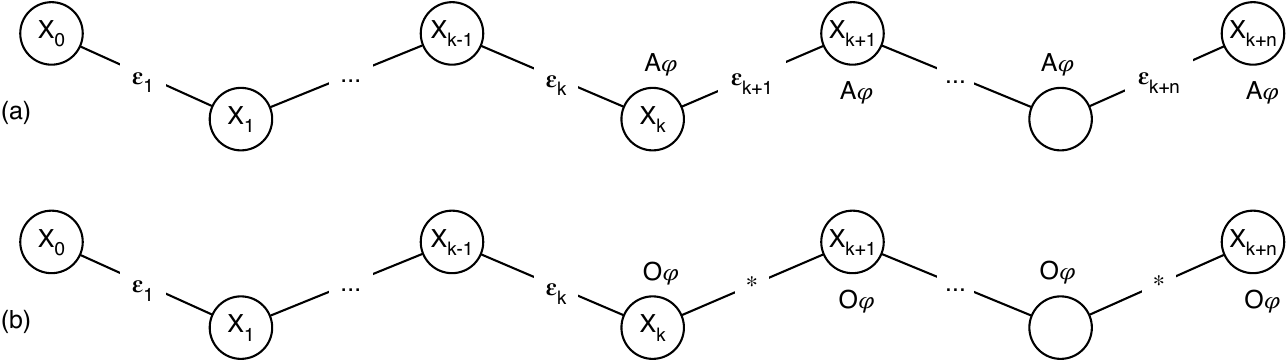}}

\vspace{0mm}
\caption{Illustrations: (a) for Lemma~\ref{boxdot up}, and (b) for Lemma~\ref{box up} and Lemma~\ref{box down}.}\label{equivalence figure}
\vspace{-4mm}
\end{center}
\vspace{0mm}
\end{figure}

\begin{lemma}\label{boxdot up}
For any $k\ge 0$, any $n\ge 0$, and any 
$(X_0,\epsilon_1,X_1,\dots,\epsilon_k,X_k,\epsilon_{k+1},X_{k+1},\dots,\epsilon_{k+n},X_{k+n})\in W_\infty,$
if $\boxdot\phi\in X_k$, then $\boxdot\phi\in X_{k+n}$.
\end{lemma}
\begin{proof}
We prove the lemma by induction on $n$, see Figure~\ref{equivalence figure} (a). If $n=0$, then assumption $\boxdot\phi\in X_k$ implies that $\boxdot\phi\in X_{k+n}$. If $n>0$, then $\boxdot\phi\in X_{k+n-1}$ by the induction hypothesis. Hence, $X_{k+n-1}\vdash\boxdot\boxdot\phi$ by the  Positive Introspection axiom.

\noindent{Case I:} $\epsilon_{k+n}=\ast$. Note that $X_{k+n-1}\vdash\boxdot\boxdot\phi$ implies  $X_{k+n-1}\vdash\Box\boxdot\phi$ by the Monotonicity axiom. Hence, $\Box\boxdot\phi\in X_{k+n-1}$ due to the maximality of set $X_{k+n-1}$. Then, $\boxdot\phi\in X_{k+n}$ by Definition~\ref{worlds and evidence definition} and due to the assumption $\epsilon_{k+n}=\ast$.

\noindent{Case II:} $\epsilon_{k+n}\neq\ast$. Statement $X_{k+n-1}\vdash\boxdot\boxdot\phi$ implies that $\boxdot\boxdot\phi\in X_{k+n-1}$ due to the maximality of set $X_{k+n-1}$. Thus, $\boxdot\phi\in X_{k+n}$ by Definition~\ref{worlds and evidence definition}.
\end{proof}

Next, we show that if $hd(w)$ contains a formula $\Box\phi$, then so does $hd(u)$ for each descendant $u$ of vertex $w$ reachable through edges all of which are labeled by symbol $\ast$.

\begin{lemma}\label{box up}
For any $k\ge 0$ and $n\ge 0$, and any 
$$(X_0,\epsilon_1,X_1,\dots,\epsilon_k,X_k,\epsilon_{k+1},X_{k+1},\dots,\epsilon_{k+n},X_{k+n})\in W_\infty,$$
if $\Box\phi\in X_k$ and $\epsilon_i=\ast$ for all $i$ such that $k<i\le k+n$, then $\Box\phi\in X_{k+n}$.
\end{lemma}
\begin{proof}
We prove the lemma by induction on $n$, see Figure~\ref{equivalence figure} (b). If $n=0$, then assumption $\Box\phi\in X_k$ implies that $\Box\phi\in X_{k+n}$. If $n>0$, then $\Box\phi\in X_{k+n-1}$ by the induction hypothesis. Hence, $X_{k+n-1}\vdash\Box\Box\phi$ by Lemma~\ref{potential positive introspection lemma}. Thus, $\Box\phi\in X_{k+n}$ by Definition~\ref{worlds and evidence definition} and due to the assumption $\epsilon_{k+n}=\ast$.
\end{proof}

The next lemma is a converse of Lemma~\ref{box up}. It shows that if $hd(u)$ contains a formula $\Box\phi$, where $u$ is a descendant of vertex $w$ reachable through edges all of which are labeled by symbol $\ast$, then so does $hd(w)$.

\begin{lemma}\label{box down}
For any $k\ge 0$, any $n\ge 0$, and any 
$$(X_0,\epsilon_1,X_1,\dots,\epsilon_k,X_k,\epsilon_{k+1},X_{k+1},\dots,\epsilon_{k+n},X_{k+n})\in W_\infty,$$
if $\Box\phi\in X_{k+n}$ and $\epsilon_i=\ast$ for all $i$ such that $k<i\le k+n$, then $\Box\phi\in X_k$.
\end{lemma}
\begin{proof}
We prove the statement of the lemma by induction on $n$, see again Figure~\ref{equivalence figure} (b). If $n=0$, then assumption $\Box\phi\in X_{k+n}$ implies that $\Box\phi\in X_k$. In what follows we assume that $n>0$.

\noindent{Case I:} $\Box\phi\notin X_{k+n-1}$. Thus, $\neg\Box\phi\in X_{k+n-1}$ due to the maximality of the set $X_{k+n-1}$. Hence, $X_{k+n-1}\vdash \Box\neg\Box\phi$ by  the Negative Introspection axiom. Then, $\Box\neg\Box\phi\in X_{k+n-1}$ due to the maximality of the set $X_{k+n-1}$. Hence, $\neg\Box\phi\in X_{k+n}$ by Definition~\ref{worlds and evidence definition} and because $\epsilon_{k+n}=\ast$. Thus, $\Box\phi\notin X_{k+n}$ due to the consistency of the set $X_{k+n}$, which is a contradiction with the assumption of the lemma.

\noindent{Case II:} $\Box\phi\in X_{k+n-1}$. Thus, $\Box\phi\in X_k$ by the induction hypothesis.
\end{proof}

The next two lemmas are relatively standard lemmas for a completeness proof of a modal logic. Their proofs show how a sequence representing an epistemic world can be extended to produce different types of child nodes on the trees in Figure~\ref{forest figure}.

\begin{lemma}\label{diamond dot lemma}
For any $w\in W_\infty$, any $\neg\boxdot\phi\in hd(w)$, and any finite $F\subseteq W_\infty$, there is $u\in W_\infty$ such that $w\sim_F u$ and $\neg\phi\in hd(u)$. 
\end{lemma}
\begin{proof}
We first show that the following set is consistent:
$
Y_0=\{\neg\phi\}\cup\{\psi\;|\; \boxdot\psi\in hd(w)\}.
$
Assume the opposite. Thus, there must exist $\boxdot\psi_1, \dots, \boxdot\psi_n\in hd(w)$ such that
$
\psi_1,\dots,\psi_n\vdash\phi.
$
Hence, by the deduction theorem for propositional logic,
$
\vdash\psi_1\to(\psi_2\to\dots(\psi_n\to\phi)\dots).
$
Then, by the  Necessitation inference rule,
$
\vdash\boxdot(\psi_1\to(\psi_2\to\dots(\psi_n\to\phi)\dots)).
$
By the  Distributivity axiom and the Modus Ponens inference rule,
$
\vdash\boxdot\psi_1\to\boxdot(\psi_2\to\dots(\psi_n\to\phi)\dots)).
$
By the Modus Ponens inference rule,
$
\boxdot\psi_1\vdash\boxdot(\psi_2\to\dots(\psi_n\to\phi)\dots)).
$
By repeating the last two steps $n-1$ times,
$
\boxdot\psi_1,\dots,\boxdot\psi_n\vdash\boxdot\phi.
$
Hence, $hd(w)\vdash \boxdot\phi$ by the choice of formulae $\psi_1,\dots,\psi_n$. Thus, $\neg\boxdot\phi\notin hd(w)$ due to the consistency of the set $hd(w)$, which  contradicts the assumption of the lemma. Therefore, set $Y_0$ is consistent. 

Let $Y$ be any maximal consistent extension of set $Y_0$ and let $u$ be the sequence $w\!::\!(F,Y)$.
We next show that $u \in W_\infty$. Indeed, since $w\in W_\infty$, by Definition~\ref{infty definition}, there must exist $n_1\ge 0$ such that $w\in W_{n_1}$. At the same time, since set $F$ is a {\em finite} subset of $W_\infty$, by Definition~\ref{infty definition}, there must exist $n_2\ge 0$ such that $F\subseteq \bigcup_{i\le n_2}W_i$. Let $n=\max\{n_1,n_2\}$. Thus, $w\in \bigcup_{i\le n}W_i$ and $F\subseteq\bigcup_{i\le n}W_i$. Hence, $w\!::\!(F,Y)\in W_{n+1}$ by Definition~\ref{worlds and evidence definition}. Therefore, $u=w\!::\!(F,Y)\in W_\infty$ by Definition~\ref{infty definition}.  
Finally, $w\sim_F u$ by Definition~\ref{sim}. 
To finish the proof of the lemma, note that $\neg\phi\in Y_0\subseteq Y=hd(u)$.
\end{proof}

\begin{lemma}\label{diamond lemma}
For any  $w\in W_\infty$ and any $\neg\Box\phi\in hd(w)$, there is $u\in W_\infty$ such that $w\sim_{W_\infty} u$ and $\neg\phi\in hd(u)$. 
\end{lemma}
\begin{proof}
We first show that the following set is consistent:
$
Y_0=\{\neg\phi\}\cup\{\psi\;|\; \Box\psi\in hd(w)\}.
$
Assume the opposite. Thus, there must exist formulae $\Box\psi_1, \dots, \Box\psi_n\in hd(w)$ and
$
\psi_1,\dots,\psi_n\vdash\phi.
$
Hence, by the deduction theorem for the propositional logic,
$
\vdash\psi_1\to(\psi_2\to\dots(\psi_n\to\phi)\dots).
$
Then, by the  Necessitation inference rule,
$
\vdash\boxdot(\psi_1\to(\psi_2\to\dots(\psi_n\to\phi)\dots)).
$
By the Monotonicity axiom and the Modus Ponens Inference rule,
$
\vdash\Box(\psi_1\to(\psi_2\to\dots(\psi_n\to\phi)\dots)).
$
By  the Distributivity axiom and the Modus Ponens inference rule,
$
\vdash\Box\psi_1\to\Box(\psi_2\to\dots(\psi_n\to\phi)\dots).
$
By the Modus Ponens inference rule,
$
\Box\psi_1\vdash\Box(\psi_2\to\dots(\psi_n\to\phi)\dots).
$
By repeating the last two steps $n-1$ times,
$
\Box\psi_1,\dots,\Box\psi_n\vdash\Box\phi.
$
Hence, $hd(w)\vdash \Box\phi$ by the choice of formulae $\psi_1,\dots,\psi_n$. Thus, $\neg\Box\phi\notin hd(w)$ due to the consistency of the set $hd(w)$, which contradicts the assumption of the lemma. Therefore, set $Y_0$ is consistent. 

Let $Y$ be any maximal consistent extension of set $Y_0$ and let $u$ be sequence $w\!::\!(\ast,Y)$. We next show that $u \in W_\infty$.  Indeed, since $w\in W_\infty$, by Definition~\ref{infty definition}, there must exist $n\ge 0$ such that $w\in W_{n}$. Hence, $w\!::\!(\ast,Y)\in W_{n+1}$ by Definition~\ref{worlds and evidence definition}. Therefore, $u=w\!::\!(\ast,Y)\in W_\infty$ by Definition~\ref{infty definition}.

Finally, let us observe that $w\sim_e u$ for each $e\in W_\infty$ by Definition~\ref{sim}. Thus, $w\sim_{W_\infty} u$.
To finish the proof of the lemma, note that $\neg\phi\in Y_0\subseteq Y=hd(u)$.
\end{proof}

By Definition~\ref{worlds and evidence definition}, elements of set $W_\infty$ are sequences of the form $(X_0,\epsilon_1,X_1,\dots,\epsilon_k,X_k)$,  where $\epsilon_i$ is either symbol $\ast$ or a subset of $W_\infty$. One might wonder if elements of $W_\infty$ are wellfounded. In other words, is it possible for an element $w\in W_\infty$ to be a member of one of its own $\epsilon_i$? Lemma~\ref{well-foundedness lemma} shows that such elements do not exist. This is a very important observation for our proof of completeness. 
Lemma~\ref{no time machine lemma} is essentially a different form of Lemma~\ref{well-foundedness lemma} which is easier to prove by induction. 

\begin{lemma}\label{no time machine lemma}
For any $k\ge 0$, any $n\ge 0$, any $t\ge 0$, and any $w=(X_0,\epsilon_1,X_1,\dots,\epsilon_k,X_k,\epsilon_{k+1},X_{k+1},\dots$, $\epsilon_{k+n},X_{k+n})\in W_t$, if $\epsilon_k\neq \ast$, then $\epsilon_k\subseteq \bigcup_{i=0}^{t-1}W_i$.
\end{lemma}
\begin{proof}
We prove this statement by induction on $n$. First, let $n=0$. Thus, $w=(X_0,\epsilon_1,X_1,\dots,\epsilon_k,X_k)$. Hence, By Definition~\ref{worlds and evidence definition}, assumptions $w\in W_t$ and $\epsilon_k\neq \ast$ imply that $\epsilon_k\subseteq \bigcup_{i=0}^{t-1}W_i$.

Suppose now that $n>0$. By Definition~\ref{worlds and evidence definition}, the assumption $w\in W_t$  implies that $$(X_0,\epsilon_1,X_1,\dots,\epsilon_k,X_k,\epsilon_{k+1},X_{k+1},\dots,\epsilon_{k+n-1},X_{k+n-1})\in W_{t-1}.$$
Thus, by the induction hypothesis, $\epsilon_k\subseteq \bigcup_{i=0}^{t-2}W_i$.
Therefore, $\epsilon_k\subseteq \bigcup_{i=0}^{t-1}W_i$.
\end{proof}



\begin{lemma}\label{well-foundedness lemma}
$w\notin \epsilon_k$ for each $w=(X_0,\epsilon_1,X_1,\dots,\epsilon_n,X_n)\in W_\infty$ and each  $k\le n$.
\end{lemma}
\begin{proof}
By Definition~\ref{infty definition}, assumption $w\in W_\infty$ implies that there is $t\ge 0$ such that $w\in W_t$. Let $m$ be the smallest $m$ such that $w\in W_m$. Thus, $w\notin \bigcup_{i=0}^{m-1}W_i$. 
At the same time, $\epsilon_k\subseteq\bigcup_{i=0}^{m-1}W_i$ by Lemma~\ref{no time machine lemma}. Therefore, $w\notin \epsilon_k$. 
\end{proof}

The next lemma puts together the pieces of the proof that we have developed. It connects the satisfaction of a formula in an epistemic world of the canonical evidence model with the maximal consistent sets out of which the world is constructed.

\begin{lemma}\label{iff lemma}
$w\Vdash\phi$ iff $\phi\in hd(w)$, for each epistemic world $w\in W_\infty$ and each formula $\phi\in\Phi$.
\end{lemma}
\begin{proof}
We prove the lemma by induction on the structural complexity of formula $\phi$. If formula $\phi$ is an atomic proposition, then the required follows from Definition~\ref{pi definition} and Definition~\ref{sat}. If formula $\phi$ is a negation or an implication, then the required follows from Definition~\ref{sat} and the maximality and the consistency of the set $hd(w)$ in the standard way. 

Suppose now that formula $\phi$ has the form $\boxdot\psi$.

\noindent{$(\Rightarrow)$} If $\boxdot\psi\notin hd(w)$, then $\neg\boxdot\psi\in hd(w)$ due to the maximality of the set $hd(w)$. To prove $w\nVdash\boxdot\psi$, by Definition~\ref{sat}, we need to show that for any finite set $F\subseteq W_\infty$ there is $u\in W_\infty$ such that $w\sim_F u$ and $u\nVdash\psi$. Indeed, by Lemma~\ref{diamond dot lemma}, there is $u\in W_\infty$ such that $w\sim_F u$ and $\neg\psi\in hd(u)$. Thus, $\psi\notin hd(u)$ due to the consistency of the set $hd(u)$. Therefore, $u\nVdash\psi$ by the induction hypothesis.

\noindent{$(\Leftarrow)$} Suppose that $\boxdot\psi\in hd(w)$. Thus, $hd(w)\vdash\boxdot\boxdot\psi$ by the Positive Introspection axiom. Hence, $hd(w)\vdash\Box\boxdot\psi$ by the Monotonicity axiom. Then, 
\begin{equation}\label{box boxdot eq}
    \Box\boxdot\psi\in hd(w)
\end{equation}
due to the  maximality of the set $hd(w)$.

\begin{figure}[ht]
\begin{center}
\vspace{0mm}
\scalebox{0.7}{\includegraphics{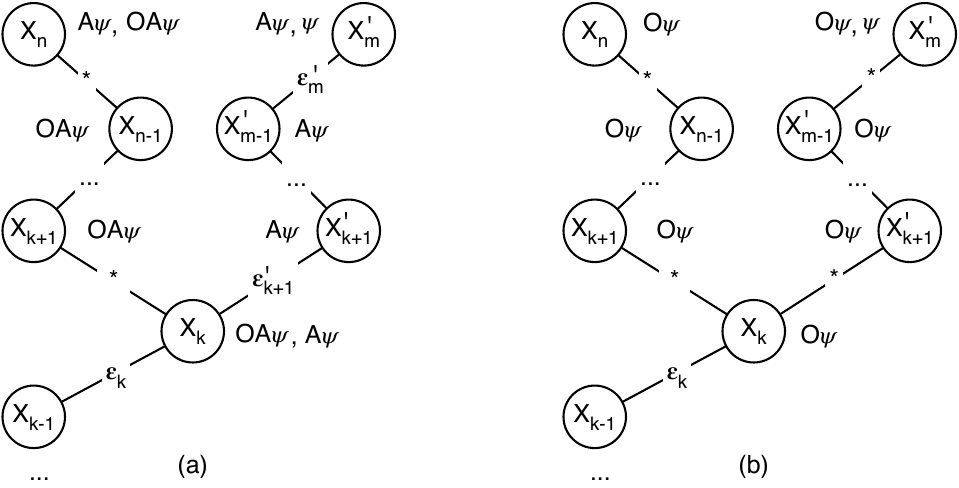}}
\vspace{-2mm}
\caption{Illustrations for the proof of Lemma~\ref{iff lemma}. }\label{iff lemma figure}
\vspace{-4mm}
\end{center}
\vspace{0mm}
\end{figure}

By Definition~\ref{sat}, it suffices to show that $u\Vdash\psi$ for all $u\in W_\infty$ such that $w\sim_{w}u$. By Definition~\ref{sim}, assumption $w\sim_{w}u$ implies that if
$w=(X_0,\epsilon_1,X_1,\dots,\epsilon_n,X_n)$
and
$u=(X'_0,\epsilon'_1,X'_1,\dots,\epsilon'_m,X'_m)$,
then there is $k$ such that, 
\begin{enumerate}
    \item $0\le k\le \min\{n,m\}$,
    \item $X_i=X'_i$ for all $i$ such that $0\le i\le k$,
    \item $\epsilon_i=\epsilon'_i$ for all $i$ such that $0< i\le k$,
    \item for all $i$, if $k<i\le n$, then either $w\in \epsilon_i$ or $\epsilon_i=\ast$,
    \item for all $i$, if $k<i\le m$, then either $w\in \epsilon'_i$ or $\epsilon'_i=\ast$.
\end{enumerate}
By Lemma~\ref{well-foundedness lemma}, $w\notin \epsilon_i$ for each $i\le n$. Thus, from the condition 4 above, $\epsilon_i=\ast$ for each $i$ such $k<i\le n$, see Figure~\ref{iff lemma figure} (a).
Hence, by Lemma~\ref{box down}, statement (\ref{box boxdot eq}) implies that $\Box\boxdot\psi\in X_k$. Thus, $X_k\vdash\boxdot\psi$ by the  Truth axiom. Then, $\boxdot\psi\in X_k$ due to the maximality of set $X_k$. Therefore, $\psi\in X'_m=hd(u)$ by Lemma~\ref{boxdot up} and condition 5 above. Therefore, $u\Vdash\psi$ by the induction hypothesis.

Finally, let formula $\phi$ have the form $\Box\psi$.

\noindent{$(\Rightarrow)$} If $\Box\psi\notin hd(w)$, then $\neg\Box\psi\in hd(w)$ due to the maximality of the set $hd(w)$. To prove $w\nVdash\Box\psi$, by Definition~\ref{sat}, we need to show that $u\nVdash\psi$ for some $u\in W_\infty$ such that $w\sim_{W_\infty} u$. Indeed, by Lemma~\ref{diamond lemma}, there is $u\in W_\infty$ such that $w\sim_{W_\infty} u$ and $\neg\psi\in hd(u)$. Thus, $\psi\notin hd(u)$ due to the consistency of the set $hd(u)$. Therefore, $u\nVdash\psi$ by the induction hypothesis.

\noindent{$(\Leftarrow)$} Suppose that $\Box\psi\in hd(w)$. By Definition~\ref{sat}, it suffices to show that $u\Vdash\psi$ for all $u\in W_\infty$ such that $w\sim_{W_\infty}u$. Indeed, let
$w=(X_0,\epsilon_1,X_1,\dots,\epsilon_n,X_n)$
and
$u=(X'_0,\epsilon'_1,X'_1,\dots,\epsilon'_m,X'_m)$. 
By Definition~\ref{sim}, assumption $w\sim_{W_\infty}u$ implies that for each $e\in W_\infty$ there is integer $k_e$ such that
\begin{enumerate}
    \item $0\le k_e\le \min\{n,m\}$,
    \item $X_i=X'_i$ for all $i$ such that $0\le i\le k_e$,
    \item $\epsilon_i=\epsilon'_i$ for all $i$ such that $0< i\le k_e$,
    \item for all $i$, if $k_e<i\le n$, then either $w\in \epsilon_i$ or $\epsilon_i=\ast$,
    \item for all $i$, if $k_e<i\le m$, then either $w\in \epsilon'_i$ or $\epsilon'_i=\ast$.
\end{enumerate}
By Lemma~\ref{infinity is infinite}, set $W_\infty$ is infinite and, thus, it is nonempty. Hence, set $\{k_e\;|\;e\in W_\infty\}$ is nonempty. Set  $\{k_e\;|\;e\in W_\infty\}$ is finite because $0\le k_e\le \min\{n,m\}$ for each $e\in W_\infty$. Let $k$ be the maximal element of the set $\{k_e\;|\;e\in W_\infty\}$. Hence, 
\begin{enumerate}
    \item $0\le k\le \min\{n,m\}$,
    \item $X_i=X'_i$ for all $i$ such that $0\le i\le k$,
    \item $\epsilon_i=\epsilon'_i$ for all $i$ such that $0< i\le k$,
    \item for all $i$, if $k<i\le n$, then either $W_\infty\subseteq \epsilon_i$ or $\epsilon_i=\ast$,
    \item for all $i$, if $k<i\le m$, then either $W_\infty\subseteq \epsilon'_i$ or $\epsilon'_i=\ast$.
\end{enumerate}
By Lemma~\ref{infinity is infinite}, set $W_\infty$ is infinite yet sets $\epsilon'_i$ and $\epsilon_i$ are finite for each $i$ by Definition~\ref{worlds and evidence definition}. Hence, $W_\infty\nsubseteq \epsilon_i$  and $W_\infty\nsubseteq \epsilon'_i$ for each $i$. Thus, conditions 4 and 5 above imply that $\epsilon_i=\ast$ for all $i$ such that $k<i\le n$ and $\epsilon_i=\ast$ for all $i$ such that $k<i\le n$, see Figure~\ref{iff lemma figure} (b). Thus, by Lemma~\ref{box down}, assumption $\Box\psi\in hd(w)=X_n$ implies that $\Box\psi\in X_{k}=X'_{k}$. Therefore,  $\psi\in X'_m=hd(u)$, by Lemma~\ref{box up}.
\end{proof}

\vspace{1mm}

We are now ready to finish the proof of Theorem~\ref{completeness theorem}. Set $X\cup \{\neg\phi\}$ is consistent by the assumption $X\nvdash\phi$ of the theorem. Consider any maximal consistent extension $X_0$  of set $X\cup \{\neg\phi\}$. Let $w_0$ be the single-element sequence $(X_0)$. Thus,  $w_0\Vdash\chi$ for each formula $\chi\in X$ and $w_0\nVdash\phi$ by Lemma~\ref{iff lemma}. This concludes the proof of the theorem.

\section{Conclusion}\label{conclusion section}

In this paper we have shown that knowledge obtained from infinitely many pieces of evidence has properties captured by modal logic S5 and knowledge obtained from finitely many pieces of evidence has properties described by modal logic S4. The main technical result is a sound and complete propositional bi-modal logic that captures properties of both of these types of knowledge and their interplay. 

A natural next step is to consider first-order logic with the same two modalities. Note that the knowledge modality $\Box$ satisfies Barcan Formula~\cite{b46jsl} $\forall x\,\Box\phi\to\Box\,\forall x\, \phi$ because this formula is derivable from S5 axioms stated in the first-order modal language~\cite{p56jsl}. At the same time, attainable knowledge modality $\boxdot$ does not satisfy Barcan Formula $\forall x\,\boxdot\phi\to\boxdot\,\forall x\, \phi$. Indeed, if variable $x$ ranges over an infinite domain and each value of $x$ in this domain has a distinct single evidence $e_x$ that justifies $\phi(x)$, then  $\forall x\,\boxdot\phi$ is true, but $\boxdot\,\forall x\, \phi$ is not. A complete axiomatization of the interplay of these two modalities in the first-order language remains an open problem.

\vspace{0mm}

\bibliographystyle{eptcs}
\bibliography{sp}

\end{document}